\documentclass[10pt]{article}

\usepackage[margin=1in]{geometry}
\usepackage[cmex10]{amsmath}
\usepackage{enumerate}
\usepackage{graphicx}
\usepackage{amsfonts}
\usepackage{amsthm}
\usepackage{url}
\usepackage{refcount}

\newcommand{\wcapacity}{\textsc{Mwisl}}
\newcommand{\capswcapacity}{\textsc{MWISL}}

\newcommand{\cG}{{\cal G}}

\newcommand{\cE}{{\cal E}}

\newcommand{\cI}{{\cal I}}
\newcommand{\cS}{{\cal S}}
\newcommand{\cP}{{\cal P}}
\newcommand{\cF}{{\cal F}}
\newcommand{\cC}{{\cal C}}
\newcommand{\cGm}{{{\cal G}_{MC-MA}}}
\newcommand{\s}{{\mathfrak l}}
\newcommand{\e}{{\mathfrak e}}
\newcommand{\Ds}{{\Delta}}

\newtheorem{proposition}{Proposition}

\newtheorem{lemma}{Lemma}

\newtheorem{theorem}{Theorem}

\newtheorem{definition}[proposition]{Definition}

\begin{document}

\title{Universal Framework for Wireless Scheduling Problems}

\author{
  Magn\'us M. Halld\'orsson
  \qquad
  Tigran Tonoyan \\ \\
  ICE-TCS,  Reykjavik University \\
  \url{{magnusmh,ttonoyan}@gmail.com}
}

\begin{titlepage}

\maketitle              

\begin{abstract}
An overarching issue in resource management of wireless networks is assessing their \emph{capacity}:
\emph{How much communication can be achieved in a network, utilizing all the tools available: power control, scheduling, routing, channel assignment and rate adjustment?}
We propose the first framework for approximation algorithms in the physical model that addresses these questions in full, including  rate control. The approximations obtained are doubly logarithmic in the link length and rate diversity. Where previous bounds are known, this gives an exponential improvement.

A key contribution is showing that the complex interference relationship of the physical model can be simplified into a novel type of amenable conflict graphs, at a small cost. 
We also show that the approximation obtained is provably the best possible for any conflict graph formulation.
\end{abstract}

\thispagestyle{empty}

\end{titlepage}

\section{Introduction}

The effective use of wireless networks revolves around utilizing fully all available diversity.
This can include power control, scheduling, routing, channel assignment and
\emph{transmission rate control} on the links, the latter being an issue of key interest for us.
The long-studied topic of \emph{network capacity} deals with
how much communication can be achieved in a network when its resources are utilized to the fullest.
This can be formalized in different ways, involving a range of problems.
The communication ability of packet networks is characterized by the capacity region, i.e. the set of traffic rates that can be supported by any scheduling policy. In order to achieve \emph{low delays and optimal throughput}, the classic result of Tassiulas and Ephremides \cite{TE92} and followup work in the area (e.g.~\cite{LinShroff04}) point out a core optimization problem that lies at the heart of such questions -- the \emph{maximum weight independent set of links (\wcapacity)} problem: from a given  set of communication links with associated weights/utilities, find an \emph{independent (conflict-free, subject to the interference model) subset of maximum total weight}. This reduction applies to very general settings involving  single-hop and multi-hop, as well as fixed and controlled transmission rate networks. Moreover, approximating {\wcapacity} within any factor implies achieving the corresponding fraction of the capacity region. This makes {\wcapacity} a central  problem in the area.
Unfortunately, solving this problem in its full generality is notoriously hard, since it is well known that {\wcapacity} is effectively inapproximable (under standard complexity theory) e.g. in models described by general conflict relations or general graphs. Moreover, in general, even approximating the capacity region  in polynomial time within  a non-trivial bound, while keeping the delays in reasonable bounds,  is hard under standard assumptions~\cite{ShahTT11}.

We tackle this question in the \emph{physical model} of communication. Towards this end, we develop a general approximation framework that not only helps us to approximate {\wcapacity}, but can also be used for tackling various other scheduling problems, such as TDMA scheduling, joint routing and scheduling and others.
The problems
handled can additionally involve path or flow selection, multiple channels and radios, and packet scheduling.
We obtain \emph{double-logarithmic} (in link and rate diversity) approximation for these problems, exponentially  improving the previously known logarithmic approximations, and, importantly, extending them to incorporate \emph{different fixed rates and rate control}.
The crucial feature of the approach (which makes it so general) is that it involves transforming the complex physical model into an unweighted/undirected conflict graph and solving the problems simply on these graphs.
Perhaps surprisingly, we find that our schema attains the best possible performance of \emph{any} conflict graph representation.
Numerical simulations  show that the conflict graph framework is a good approximation for the physical model on randomly placed network instances as well.
Our approach also finesses the task of selecting optimum power settings by using \emph{oblivious} power assignment, one that depends only on the properties of the link itself and not on other links. The performance bounds are however in comparison with the optimum solution that can use arbitrary power settings.

Technically, our approach  generalizes our earlier framework \cite{us:stoc15}.
Our extensions required substantial changes throughout the whole body of arguments.
That formulation works only for uniform constant rates, and the generalization requires substantial new ideas.
One indicator of the challenges overcome is that we could prove that our
doubly-logarithmic approximation is best possible in the presence of different rates,
while better approximations are known to hold in the case of uniform rates \cite{us:stoc15}.

We make some undemanding assumptions about the settings. We assume that the networks are
\emph{interference-constrained}, in that interference rather than the ambient noise is the determining factor of proper
reception.  This assumption is common and is particularly natural in settings with rate control, since the impact of noise can always be made
negligible by avoiding the highest rates, losing only a small factor in performance.
We also assume that nodes are (arbitrarily) located in a doubling metric, which generalizes Euclidean space, allowing the modeling of some of  non-geometric effects seen in practice.

\subparagraph*{Our Results}
Our results can be summarized as follows:
\begin{itemize}
\item{We establish a general framework for tackling wireless scheduling and related problems,}
\item{Our approximations hold for nearly all such problems, including variable rates settings,}
\item{We obtain exponential improvement over previously known approximations,}
\item{The approximations are obtained via simple conflict graphs, as opposed to the complicated physical model, and by using \emph{oblivious} power assignments,}
\item{We establish tight bounds indicating the limitations of our method.}
\end{itemize}
\subparagraph*{Related work}
 Gupta and Kumar  introduced the  physical model of interference/communication with log-path fading in their influential paper~\cite{kumar00},
 and it has remained the default in analytic studies.
Moscibroda and Wattenhofer \cite{moscibrodaconnectivity} initiated worst-case analysis of scheduling problems in networks of arbitrary topology, which is also the setting of interest in this paper.
There has been significant progress in understanding scheduling problems with fixed uniform rates.
NP-completeness results have been given for different variants \cite{goussevskaiacomplexity, katz2010energy,lin2012complexity}.
Early work on approximation algorithms involve (directly or indirectly) partitioning links into length groups,
which results in performance guarantees that are at least logarithmic in
$\Delta$, the link length diversity: TDMA scheduling and uniform weights {\wcapacity} \cite{goussevskaiacomplexity,dinitz,us:talg12},
non-preemptive scheduling \cite{fu2009power}, joint power control, scheduling and routing \cite{chafekarcrosslayer},
and joint power control, routing and throughput scheduling in multiple channels \cite{AG10}, to name a few.
Constant-factor approximations are known for uniform weight {\wcapacity} (in restricted power regimes \cite{SODA11}
and (general) power control \cite{kesselheimconstantfactor}).
Standard approaches translate the constant-factor approximations for the uniform weight {\wcapacity}  into $O(\log n)$ approximations for TDMA scheduling  and general {\wcapacity}.
Many problems become easier, including {\wcapacity} and TDMA scheduling, in the regime of linear power assignments
\cite{fangkeslinear,wang2011constant,halmitcognitive,tonoyanlinear}.
Recently, a $O(\log^* \Ds)$-approximation algorithm was given for TDMA scheduling and {\wcapacity} \cite{us:stoc15}, by transforming the physical model into a conflict graph. We build on this approach, and extend it into a general framework that covers other problems and incorporates support for rate control.

Very few results are known for problems involving rate control.
The constant-factor approximation for  {\wcapacity} with uniform weights and arbitrary but fixed rates proposed by Kesselheim~\cite{KesselheimESA12} can be used to obtain  $O(\log n)$-approximations for TDMA scheduling and {\wcapacity} with rate control, where $n$ is the number of links.
Another recent work~\cite{goussevskaia2016wireless} handles the TDMA scheduling problem (with fixed but different rates), obtaining an approximation independent of the number of links $n$, but the ratio is polynomial in $\Ds$.
There have been numerous algorithms that try to approximate or replace {\wcapacity} in the context of packet scheduling. Several examples include Longest-Queue-First Scheduling (LQF)~\cite{JooLS09},  Maximal Scheduling~\cite{WuSP07}, Carrier Sense Multiple Access (CSMA)~\cite{jiang2010distributed}. The approximations obtained usually depend on some parameter of the conflict graph, such as the \emph{interference degree}. In the case of CSMA (and  other similar protocols), it is known that the algorithms are throughput-optimal, but in general they take exponential time to stabilize, or otherwise require constant degree conflict graphs~\cite{JiangLNSW12}.
It is also well known that many scheduling problems such as vertex coloring and {\wcapacity} are easy to approximate in bounded \emph{inductive independence} graphs, such as geometric intersection graphs or protocol model. However, fidelity to the cumulative nature of interference  and the question of modeling rate control are among the significant issues faced by such graph models.

\subparagraph*{Paper Organization} The fundamental ideas of our approximation framework are described in Sec.~\ref{S:approxmethod}. After introducing the model and definition in Sec.~\ref{S:model}, we derive
the core technical part, the approximation of the physical model by the conflict graphs,
 in Sec.~\ref{S:approximation}, and the optimality of approximation. The framework is applied to obtain approximations for  fixed  rate scheduling problems in Sec.~\ref{S:fixedrate} and for  problems with rate control in Sec.~\ref{S:ratecontrol} (the latter two can be read separately from Sec.~\ref{S:approximation}).
Due to space constraints, several technical proofs  are deferred to the appendix.

\section{Approximation Method}\label{S:approxmethod}

Before defining the details, let us describe the main idea behind the approximation technique. In essence, we define a notion of approximation of an independence system\footnote{An independence system $\cI$ over a set of vertices $V$  is a pair $\cI=(V,\cE)$, where $\cE\subseteq 2^V$ is a collection of subsets of vertices that is closed under subsetting: if $S\in \cE$ and $S' \subset S$, then $S' \in \cE$.} $\cI_\cP = (L, \cE_\cP)$  by a graph  $\cG  = (L, E)$ over the set $L$ of links. The system  $\cI_\cP$ corresponds to the cumulative interference in the physical model, while $\cG$ is a conflict graph  describing pairwise conflicts between links. We will refer to independent sets in $\cI_\cP$ as \emph{feasible sets}, and to independent sets in $\cG$ as \emph{independent sets}, to avoid confusion.

The approximation is described by several key properties.

\textbf{Refinement (Feasibility of Independent Sets).} Every independent set $S$ in $\cG$ must be feasible, i.e. $S \in \cE_\cP$. Thus, finding an independent set in $\cG$ gives also a feasible set in $\cI_\cP$.

\textbf{Tightness (of refinement).} There is a small number $k$ such that every feasible set $S\in \cI_\cP$ is a union of at most $k$ independent sets in $\cG$. The smallest such $k$ is called the \emph{tightness} of refinement. This property guarantees that even an optimal (for a problem in question) feasible set can be covered with a few independent sets.

The two properties above establish a tight connection between the two models. That allows us to take nearly \emph{every} scheduling problem in the physical model and reduce it to the corresponding problem in conflict graphs (in a way formalized in Sec.~\ref{S:framework}), by paying only an approximation factor depending on the tightness $k$. However, in order for this scheme to work, it should be easier to solve such problems in $\cG$, which leads to the third key property.

\textbf{Computability.} There are efficient (approximation) algorithms for scheduling-related problems such as \emph{vertex coloring} and \emph{maximum weight independent set} in $\cG$.

A graph $\cG$ satisfying the properties above is said to be a \emph{refinement} of $\cI_{\cP}$.
The main effort in the following two sections is to define an appropriate conflict graph refinement for the physical model and prove these key properties. We find such a  family that approximates the physical model with nearly constant tightness, i.e. double-logarithmic in length and rate diversity and show that this is best possible for any conflict graph, up to constant factors.
This approximation allows us to bring to bear the large body of theory of graph algorithms, greatly simplifying both the exposition and the analysis.

\section{Model}
\label{S:model}

In scheduling problems, the basic object of consideration is a set $L$ of $n$ communication links, numbered from $1$ to $n$, where each link  $i\in L$ represents a single-hop communication request between two wireless nodes located in a metric space -- a sender node $s_i$ and  receiver node $r_i$.

We assume the nodes are located in a metric space with distance function $d$\label{G:distance}.
We denote $d_{ij}=d(s_i,r_j)$\label{G:asymdistance} and $l_i=d(s_i,r_i)$\label{G:li}.
The latter is called the \emph{length} of link $i$. Let $d(i,j)$ denote the minimum distance between the nodes of links $i$ and $j$.

The nodes have adjustable transmission power levels.
A \emph{power assignment} for the set $L$ is a function $P:L\rightarrow \mathbb{R}_+$. For each link $i$, $P(i)$\label{G:power} defines the power level used by the sender node $s_i$.
In the physical model of communication, when using a power assignment $P$, a transmission of a link $i$ is successful if and only if
\begin{equation}\label{E:sinr}
SIR(S,i)=\frac{P(i)/l_i^{\alpha}}{\sum_{j\in S\setminus \{i\}} P(j)/d_{ji}^{\alpha}}\ge \beta_i,
\end{equation}
where  $\beta_i>1$\label{G:beta} denotes the minimum signal to noise ratio required for link $i$, $\alpha \in (2,6)$\label{G:alpha} is the path loss exponent and $S$ is the set of links transmitting concurrently with link $i$.
Note that we omit the noise term in the formula above, since we focus on interference limited networks. This can be justified by the fact that one can simply slightly decrease the transmission rates to make the effect of the noise negligible, then restore the rates by paying only constant factors in  approximation.

A set $S$ of links is called $P$-\emph{feasible} if the condition~(\ref{E:sinr}) holds for each link $i\in S$ when using power assignment $P$. We say $S$ is \emph{feasible} if there exists a power assignment $P$ for which $S$ is $P$-feasible.

\textbf{Effective Length.}
Let us denote $\s_i=\beta_i^{1/\alpha} l_i$  and call it the \emph{effective length} of link $i$. Let $\Ds(L)=\max_{i,j\in L}\{\s_i/\s_j\}$ denote the \emph{(effective) length diversity} of  instance $L$.
We call a set $S$ of links \emph{equilength} if for every two links $i,j\in S$, $\s_i \le 2\s_j$, i.e., $\Ds(S) \le 2$.
Note that with the introduction of effective length, the feasibility constraint~(\ref{E:sinr}) becomes: $\frac{P(i)}{\s_i^\alpha} \ge \sum_{j\in S\setminus \{i\}}\frac{P(j)}{d_{ji}^{\alpha}}$. This looks like the same formula but with uniform rates $\beta_i'=1$. However, there is an essential difference between the two: the quantities $\s_i$ are not related to the metric space in the same way as   lengths $l_i$, as $\s_i$ can be arbitrarily larger than $l_i$.

\textbf{Metrics.} The \emph{doubling dimension} of a metric space is the infimum of all numbers $\delta > 0$
such that for every $\epsilon$, $0 < \epsilon \le 1$, every ball of radius $r>0$ has at most $C\epsilon^{-\delta}$
points of mutual distance at least $\epsilon r$ where $C\geq 1$ is an absolute constant.
For example, the $m$-dimensional Euclidean space
has doubling dimension $m$~\cite{heinonen}.  We let $m$\label{G:dimension} denote the doubling dimension of the space
containing the links. We will assume $\alpha>m$,  which is the standard assumption $\alpha >2$ in the Euclidean plane.

\section{Conflict Graph Approximation of Physical Model}
\label{S:approximation}

In this section we present the $O(\log\log\Ds)$-tight refinement of the physical model by conflict graphs. The first part introduces our conflict graph $\cG_f$ that generalizes  the conflict graph definition of~\cite{us:stoc15} and  extends it to general thresholds/rates. The three subsequent parts give the proofs of the three key properties: refinement, tightness and computability. The last part argues the asymptotic optimality of $O(\log\log\Ds)$-tightness for any conflict graph, which contrasts the $O(\log^*\Ds)$ bound known in the uniform thresholds setting.
\begin{theorem}
There is an $O(\log\log\Ds)$-tight refinement of the physical model by a conflict graph family $\cG(L)$.
\end{theorem}

\subparagraph*{Conflict Graphs} We define the conflict graph family as follows.
\begin{definition}
Let $f:\mathbb{R}_+\rightarrow \mathbb{R}_+$ be a positive non-decreasing function.
Links $i,j$ are said to be \emph{$f$-independent} if
  $ d_{ij}d_{ji} > \s_i\s_j f\left(\s_{max}/\s_{min}\right), $
where $\s_{min}=\min\{\s_i,\s_j\},\s_{max}=\max\{\s_i,\s_j\}$, and otherwise \emph{$f$-adjacent}.
A set of links is $f$-independent ($f$-adjacent) if they are pairwise $f$-independent ($f$-adjacent).

The conflict graph $\cG_f(L)$ of a set $L$ of links is the graph with vertex set $L$, where two vertices
are adjacent if and only if they are $f$-adjacent.
\end{definition}
This definition extends the conflict graphs introduced in~\cite{us:stoc15}, where the independence criterion was $d(i,j)/l_{min}>f(l_{max}/l_{min})$ ($l_{max},l_{min}$ are the length of the longer and shorter links, resp.). When all threshold values $\beta_i$ are constant, the latter essentially follows from the definition above by ``canceling'' $l_{max}$  with the larger value of $d_{ij},d_{ji}$ (modulo constant factors). In general, however, the effective lengths can be very different from the actual link lengths, and feasibility requires more separation than given by  graphs involving distances only. A technical difficulty introduced by the new definition is that we have to keep track of two distances $d_{ij}$ and $d_{ji}$ instead of the single distance $d(i,j)$, but this appears to be necessary.

We will be particularly interested in \emph{sub-linear} functions $f(x) = O(x)$.
A function $f$ is \emph{strongly sub-linear} if for each constant $c\ge 1$, there is a constant $c'$ such that $cf(x)/x\le f(y)/y$ for all $x,y\ge 1$ with $x\ge c'y$. Note that if $f$ is strongly sub-linear then $f(x)=o(x)$.
For example, the functions $f(x)=x^{\delta}$ ($\delta<1$) and $f(x)=\log{x}$ are strongly sub-linear.

\subparagraph*{Refinement: Feasibility of Independent Sets}
Our goal now is to find a function $f$ such that each independent set in $\cG_f$ is feasible. It is clear that this can be achieved by letting $f$ grow sufficiently fast. But  we should not let it grow too fast, so as to not affect tightness.
We also need to indicate which power assignment makes the independent sets in $\cG_f$ feasible.
Our approach is to  preselect a family of \emph{oblivious power assignments}, that are local to each link and do not depend on others, and then find an appropriate function $f$.
Consider the family of power assignments $P_\tau$ parameterized by $\tau\in (0,1)$, where $P_\tau(i)\sim \s_i^{\tau\alpha}$ for each link $i$.
In order to obtain $P_\tau$-feasibility, we take conflict graphs $\cG_f$ with $f(x)=\gamma x^\delta$ for $\delta\in (0,1)$ and $\gamma\ge 1$. Such graphs are denoted as $\cG_{\gamma}^\delta$.
We show that every independent set in $\cG_\gamma^\delta$ for appropriate $\gamma$ and $\delta$ is $P_\tau$-feasible for some $\tau$.

\begin{theorem}\label{T:obliviouspowers}
Let $\delta_0=\frac{\alpha-m+1}{2(\alpha-m) + 1}$. If $\delta\in (\delta_0,1)$ and the constant $\gamma>1$ is large enough, there is a value $\tau \in (0,1)$ such that each independent set in $\cG_{\gamma}^{\delta}$ is $P_{\tau}$-feasible.
\end{theorem}

The proof is an adaptation of the ideas used in the proof of~\cite[Cor. 6]{us:fsttcs15} to our definition of conflict graphs and effective lengths. Given an independent set $S$ in $\cG_{\gamma}^{\delta}$ and a link $i$, we bound the interference of $S$ on $i$ by first splitting $S$ into equilength subsets, bounding the contribution of each subset separately, then combining the bounds into one. The core of the proof is a careful application of a common packing argument in doubling metric spaces.

\subparagraph*{Tightness of Refinement}
Now, let us bound the number of $f$-independent sets that are necessary to cover a feasible set. We show that this number is $O(f^*(\Ds(S)))$ for any feasible set $S$, where $f^*$ is defined for every strongly sub-linear function, as follows.
For each integer $c\geq 1$, the function $f^{(c)}(x)$ is defined inductively by: $f^{(1)}(x)=f(x)$ and $f^{(c)}(x)=f(f^{(c-1)}(x))$\label{G:frepeated}. Let $x_0=\inf\{x\geq 1, f(x) < x\} +1$; such a point exists for every $f(x)=o(x)$. The function $f^*(x)$\label{G:fstar}, is defined by:
$
f^*(x)=\arg\min_c\{f^{(c)}(x)\le x_0\}
$ for arguments $x> x_0$, and $f^*(x)=1$ for the rest.
Note that for a function $f(x)=\gamma x^\delta$ with constants $\gamma>0$ and $\delta\in (0,1)$, $f^*(\Ds)=\Theta(\log{\log{\Ds}})$, which is the tightness bound we are aiming for.

\begin{theorem}\label{T:sandwich}
Consider a non-decreasing strongly sub-linear function $f$. Every feasible set $S$ can be split into $O(f^*(\Ds(S)))$ subsets, each independent in $\cG_{f}(S)$.
\end{theorem}

Let us fix a function $f$ with properties as in the theorem.
We establish the partition in Thm.~\ref{T:sandwich} in two steps. The first step is to show that feasible set $S$ can be partitioned into a constant number of independent sets in $\cG_{\rho}^0(S)$ for any constant $\rho$, i.e., subsets $S'$ such that for every pair of links $i,j\in S'$, $d_{ij}d_{ji} > \rho \s_i\s_j$. Such  subsets are  called $\rho$-independent for short.
The second step is to show that for an appropriate constant $\rho$, each $\rho$-independent set can be partitioned into at most $O(f^*(\Ds))$ of $f$-independent subsets.

The first step is easy. Each feasible set can be partitioned into at most $2\rho^{\alpha/2}$ subsets, each of them feasible with updated thresholds $\{\rho^{\alpha/2}\beta_i\}$.
 This is a direct application of Corollary 2 of~\cite{HB15}.
Let $S'$ be such a subset and let $i,j\in S'$. The feasibility constraint for $i$ and $j$ implies:
\[
P(i)/l_i^\alpha \ge \rho^{\alpha/2} \beta_i P(j)/d_{ji}^\alpha\text{ and }P(j)/l_j^\alpha \ge \rho^{\alpha/2}\beta_j P(i)/d_{ij}^\alpha.
\]
By multiplying together the inequalities above, canceling $P(i)$ and $P(j)$ and raising to the power of $1/\alpha$, we obtain:
$d_{ij}d_{ji} \ge \rho\s_i\s_j$, as required.

The proof of the second step requires the following lemmas, which constitute the most significant technical difference from the proof of the corresponding theorem in~\cite{us:stoc15}, as they encapsulate the technicalities of dealing with our definition of conflict graphs: It is not sufficient to bound only one of the distances between links (such as $d(i,j)$ in~\cite{us:stoc15}); we need a bound on the product of two distances.

\begin{lemma}\label{P:triangles}
Let $i,j,k$ be such that $\s_i \le \s_j\le \s_k$ and $i$ is $f$-adjacent with both $j$ and $k$, where $f$ is a non-decreasing sublinear function. Then
\[
d_{jk}d_{kj} < 18\s_i\s_kf(\s_k/\s_i) + 13\s_j\s_k
+ 2\s_j\sqrt{\s_i\s_kf(\s_k/\s_i)} + \s_k\sqrt{\s_i\s_jf(\s_j/\s_i)}.
\]
\end{lemma}

\begin{lemma}\label{P:simpleset} Let $i$ be a link and $\rho>1$.
If $E$ is a $\rho$-independent set of links where each  $j\in E$ is $f$-adjacent with $i$ and satisfies $\s_i\le \s_j \le c\s_i$ for a constant $c$, then $|E|=O(1)$.
\end{lemma}

	\begin{proof}[Proof of Theorem~\ref{T:sandwich}]
	By the discussion above, it is sufficient to show that each $\rho$-independent set $S$, for appropriate constant $\rho>1$, can be partitioned into a small number of $f$-independent sets.
	We choose $\rho=3c_f + 31$, where $c_f$ is such that $f(x)\le c_f x$ for all $x\ge 1$ (recall that $f$ is sub-linear).
		Partitioning is done by the following \emph{inductive} coloring procedure: 1. Consider the links in a non-increasing order by effective length, 2. Assign each link the first natural number that has not been assigned to an $f$-adjacent link yet.
		Clearly, such a procedure defines a partitioning of $S$ into $f$-independent subsets.

		Fix a link $i\in S$. Let $T$ denote the set of links in $j\in S$ that have greater effective length than $i$ and are $f$-adjacent with $i$. In order to complete the proof, it is enough to show that $|T|= O(f^*(\Ds))$, as $|T|$ is an upper bound on the number assigned to link $i$.

Since $f(x)$ is strongly sub-linear, there exists $x_0=\inf\{x\geq 1, f(x) < x\}+1$. Let us split $T$ into two subsets $A$ and $B$, where $A$ contains the links $j\in T$ such that $\s_j \le x_0\s_i$ and $B=T\setminus A$. By Lemma~\ref{P:simpleset}, we have that $|A|=O(1)$, so we concentrate on $B$.

Let $j,k$ be arbitrary links in $B$ such that $\s_j \le \s_k$.  By applying Lemma~\ref{P:triangles} and using the definition of $c_f$, we obtain:
	$
	d_{jk}d_{kj} < 18\s_i\s_kf(\s_k/\s_i) + (3c_f+13)\s_j\s_k.
	$
Recall that $j$ and $k$ are $(\rho=3c_f + 31)$-independent, so $d_{jk}d_{kj} > (3c_f + 31)\s_j\s_k$, which  gives us
	$
	\s_j/\s_i < f(\s_k/\s_i).
	$
	Let $1,2,\dots,t=|B|$ be an arrangement of the links in $B$ in a non-decreasing order by effective length and let $\lambda_j=\s_j/\s_i$ for  $j=1,2,\dots,t$.
 We have just shown that
	\[
	x_0\le \lambda_1 < f(\lambda_{2})\le f(f(\lambda_{3}))\le \cdots\le f^{(t-1)}(\lambda_t),
	\]
	namely, $t-1\le f^*(\lambda_t)= O(f^*(\Ds))$. Thus, $|T|=|A| + |B|=O(1) + O(f^*(\Ds))$.
		\end{proof}

\subparagraph*{Computability}
Computability of our conflict graph construction is demonstrated through the notion of \emph{inductive independence}.
An $n$-vertex graph $G$ is \emph{$k$-inductive independent} if there is an ordering  $v_1,v_2,\dots,v_n$ of vertices such that for each $v_i$, the subgraph of $G$ induced by the set $N_G(v_i)\cap \{v_i,v_{i+1},\dots,v_n\}$ has no independent set larger than $k$, where $N_G(v)$ denotes the neighborhood of vertex $v$.
It is well known, e.g.~\cite{ackoglu, yeborodin}, that vertex coloring and {\wcapacity} problems are $k$-approximable in $k$-inductive independent graphs.

\begin{theorem}\label{T:inductiveindep}
Let $f$ be a non-decreasing strongly sub-linear function with $f(x)\ge 40$ for all $x\ge 1$. For every set $L$, the graph $\cG_f(L)$  is constant inductive independent.
\end{theorem}

The proof is somewhat similar to that of Thm.~\ref{T:sandwich}. The inductive independence ordering non-decreasing order of links by length. With this in mind, the proof of Thm.~\ref{T:sandwich} can be applied, with the following core difference: while in Thm.~\ref{T:sandwich} the goal was, for a link $i$, to bound the number of \emph{ $\rho$-independent} links that have greater effective length and are $f$-adjacent with $i$, here we need to bound the number of \emph{ $f$-independent} links that have greater effective length and are $f$-adjacent with $i$.

\subparagraph*{Optimality of $O(\log\log\Ds)$-tightness}
Here we show that the obtained tightness is essentially best possible, by demonstrating that every reasonable conflict graph formulation must incur an $O(\log\log\Ds)$ factor.
We depart from some basic assumptions on conflict graphs. First, since the  feasibility of a set of links is precisely determined by the values $\s_i$ and $d_{ij}$,  we assume that in a conflict graph, the adjacency of two links $i,j$ is a predicate of  variables $\s_i,\s_j,d_{ij},d_{ji}$. Another basic observation is that the feasibility formula is scale-free w.r.t. those values; hence, we assume that so is a conflict graph formulation. This allows us to reduce the number of variables in the adjacency predicate: $\frac{\s_{max}}{\s_{min}}, \frac{d_{ij}}{\s_{min}}, \frac{d_{ji}}{\s_{min}}$, where $\s_{min}$ and $\s_{max}$ are the smaller and larger values of $\s_i,\s_j$, respectively. Our construction will consist of only \emph{unit-length} links (i.e. $l_i=1$) of mutual distance at least 3. In this case, we can further reduce the number of variables by noticing that in such instances, $d_{ij}=\Theta( d_{ji})=\Theta(d(i,j))$. Thus, the conflict relation is essentially determined by two variables: $\frac{d(i,j)}{\s_{min}}$ and $\frac{\s_{max}}{\s_{min}}$. By separating the variables, the conflict predicate boils down to a relation
$
\frac{d(i,j)}{\s_{min}} > f(\frac{\s_{max}}{\s_{min}})
$
for a function $f$. Note that this  is similar to the conflict graph definition of~\cite{us:stoc15}, except that the lengths are replaced with effective lengths.

Let us show that the refinement property requires that $f(x)=\Omega(\sqrt{x})$ in such a graph.
Let us fix a function $f:[1,\infty)\rightarrow [1,\infty)$. Let $i,j$ be unit-length links with $\beta_j=1$ and $\beta_i=X^\alpha >1$, where $X$ is a parameter. Assume further that the links $i,j$ are placed on the plane so that $d(i,j)=3f(X)=3f(\s_i/\s_j)$, which means the links are $f$-independent. Thus,  $i,j$ must form a feasible set: $\frac{P(i)}{\s_i^\alpha} > \frac{P(j)}{d_{ji}^\alpha}$ and $\frac{P(j)}{\s_j^\alpha} > \frac{P(i)}{d_{ij}^\alpha}$. Multiplying these inequalities together and canceling $P(i)$ and $P(j)$ out, gives: $d_{ij}d_{ji} > \s_i\s_j=X$. This implies that we must have $d(i,j)= \Theta(\sqrt{d_{ij}d_{ji}})=\Omega(\sqrt{X})$, which in turn implies that $f(X)=d(i,j)/3=\Omega(\sqrt{X})$.

Now, a simple modification of the construction in~\cite[Thm. 9]{us:stoc15} gives a set $S$ of unit-length links arranged on the line and with appropriately chosen thresholds $\beta_i$ and distances $d(i,j)$, such that every two links are $f$-adjacent, but the whole set $S$ is feasible. Such a construction can be done with the number of links $n=\Omega(f^*(\Ds))$, i.e. there is a \emph{feasible} set of links that cannot be split in less than $\Omega(f^*(\Ds))$ $f$-independent subsets. Since $f(x)=\Omega(\sqrt{x})$, we have  $f^*(x)=\Omega(\log\log x)$,  which proves that the tightness must be at least $\Omega(\log\log\Ds)$.

\section{Approximating Fixed-Rate Scheduling}\label{S:fixedrate}

We detail now the more classical problems that can be handled with our framework, starting with those involving fixed datarates.
Intuitively, our framework can handle a problem if there is a correspondence between solutions in the physical model instance and solutions in the refinement graph. The refinement property ensures that the graph solutions map directly to feasible solutions in the physical model --- we need to ensure a (approximate) correspondence in the other direction.
We will argue that an optimal solution in the physical model has a counterpart in the graph instance,
whose quality decreases only by the tightness factor $k$.

\subparagraph*{ General Approximation Framework}\label{S:framework}
Common scheduling-related optimization problems can be classified as \emph{covering} or \emph{packing}.

In covering  problems, a feasible solution $\sigma$ contains a (ordered) covering of the set $L$ of links with feasible sets $\pi = \langle S_1, S_2,\dots, S_t\rangle $ (i.e., $\cup_{t}S_t = L$), which we call \emph{time slots}, and the objective is to minimize a function $f_\sigma(\pi)$ of the covering, which may also depend on other problem constraints.

In packing  problems, a feasible solution $\sigma$ contains a \emph{fixed} number $c$ of feasible sets (packing),  $\eta=\langle S_1,S_2,\dots,S_c\rangle $, not necessarily covering $L$, which we call \emph{channels}, and the objective is to maximize a function $g_\sigma(\eta )$ of the packing.

Given a refinement $\cG$ and a cover $\pi=\langle S_1,S_2,\dots,S_t\rangle $ of $L$ by feasible sets, we call another cover $\pi'=\langle S_1^1,\dots S_1^{h_1},S_2^1,\dots,S_t^1,\dots,S_t^{h_t}\rangle$, a refinement of $\pi$ if $\langle S_i^1,\dots S_i^{h_i}\rangle$ is a cover of $S_i$ by \emph{independent} sets in $\cG$.
Similarly, given a packing $\eta=\langle S_1,S_2,\dots,S_c\rangle$, a refinement of $\eta$ is another packing $\eta'=\langle S'_1,S'_2,\dots,S'_c\rangle$, where $S'_i\subseteq S_i$ is an independent set in $\cG$.

Formally, a covering problem is \emph{refinable} if for every $k$-tight refinement $\cG$  and a solution $\sigma$ with cover $\pi$, there is a feasible solution $\sigma'$ containing a refinement $\pi'$ of $\pi$, and such that $f_\sigma(\pi)\ge \frac{f_{\sigma'}(\pi')}{k}.$
A packing problem is \emph{refinable} if for every $k$-tight refinement $\cG$  and a solution $\sigma$ with a packing $\eta=\langle S_1,S_2,\dots,S_c\rangle $, there is a feasible solution $\sigma'$ containing a refinement $\eta'$ of $\eta$, and such that $g_\sigma(\eta)\le k\cdot g_{\sigma'}(\eta' ).$

\begin{theorem}\label{T:refapprox}
Let $\cG$ be a $k$-tight refinement of the physical model. For every  refinable problem, a $\rho$-approximation algorithm in $\cG$ gives $k \cdot \rho$-approximation in the physical model.
\end{theorem}

Thus, in order to obtain an approximation for a specific problem, it is sufficient to show that the problem is refinable: then the solution in a $k$-tight refinement gives a solution with an additional approximation factor $k$. Refinability requires the objective function of the problem to have certain linearity property. Examples of refinable covering problems include the ones where the objective function is the number of time slots or the sum of completion times (i.e. indices of time slots). Perhaps the simplest example of a refinable packing problem is the \emph{maximal independent set of links} problem, where the objective is the size of the feasible set (i.e., there is only a single channel).
Below, we apply the refinement framework to some important scheduling problems, which leads to $O(\log\log\Ds)$-approximation for all of them.

\subparagraph*{{\capswcapacity} with Fixed Weights}
Consider the {\wcapacity} problem, where the weights $\omega_i$ of links are fixed positive numbers. It is easy to see that this is a refinable packing problem, as the objective function -- the sum of weights --  is linear with respect to partition. Thus, since there is a constant factor approximation to {\wcapacity} in $\cG(L)$ (by computability), it gives an $O(\log\log\Ds)$-approximation in the physical model (by Thm.~\ref{T:refapprox}).

\subparagraph*{Multi-Channel Selection}
Given a natural number $c$ -- the number of channels -- the goal is to select a maximum number of links that can be partitioned into $c$ feasible subsets (a subset for each channel). Again, this is easily seen to be a refinable packing problem, as the objective function -- the total number of links across all channels, is linear w.r.t. partitioning. A simple greedy algorithm gives constant factor approximation to multi-channel selection in constant-inductive independent graphs, which translates to an $O(\log\log\Ds)$-approximation in the physical model.

\subparagraph*{TDMA Scheduling}
The goal is to partition the set $L$ of links into the minimum number of feasible subsets. This is a covering problem, and the objective function is the number of slots, which is linear w.r.t. partitioning. A simple \emph{first-fit} style greedy algorithm gives constant factor approximation to vertex coloring in constant inductive independent graphs, which gives an $O(\log\log\Ds)$-approximation to TDMA scheduling in the physical model.

\subparagraph*{Fractional Scheduling}
This is a fractional variant of TDMA scheduling with an additional constraint of link demands. A \emph{fractional schedule} for a set $L$ of links is a collection of feasible sets with rational values $\cS=\{(I_k,t_k) : k=1,2\dots,q\}\subseteq \cE_{\cP}\times\mathbb{R}_+$, where $\cE_{\cP}$ is the set of all feasible subsets of $L$. The sum $\sum_{k=1}^q{t_k}$ is the \emph{length} of the schedule $\cS$. The \emph{link capacity vector} $c_{\cS}:L\rightarrow \mathbb{R}_+$ associated with the schedule $\cS$ is given by
$
c_{\cS}(i) = \sum_{(I,t)\in \cS: I\ni i}t.
$
 Essentially, the link capacity shows how much scheduling time each link gets. Finally, a \emph{link demand vector} $d:L\rightarrow \mathbb{R}_+$ indicates how much scheduling time each link needs.

The \emph{fractional scheduling problem} is a covering type problem, where given a demand vector $d$, the goal is to compute a minimum length schedule that serves the demands $d$, namely,  for each link $i\in L$,
$
c_{\cS}(i)\ge d(i).
$
Since the cost function $\sum_{k=1}^q{t_k}$ is again linear w.r.t. partitioning of a schedule, it is readily checked that the fractional scheduling problem is also refinable.
A simple greedy algorithm presented in~\cite{wan13} achieves constant factor approximation for fractional scheduling in constant inductive independent graphs. This gives an $O(\log\log\Ds)$-approximation in the physical model.

\subparagraph*{Joint Routing and Scheduling}
 Consider an ordered set of $p$ source-destination node pairs (multihop communication requests) $(u_i,v_i)$, $i=1,2,\dots,p,$ with associated weights/utilities $\omega_i>0$, in a multihop network given by a directed graph $G$, where the \emph{edges} of the graph are the transmission links. Let $\cP_i$ denote the set of directed $(u_i,v_i)$ paths in $G$ and let $\cP=\cup_i \cP_i$. Then a \emph{path flow} for the given set of requests is a set $\cF=\{(P_k,\delta_k): k=1,2,\dots\}\subseteq \cP \times \mathbb{R}_+$.  The \emph{link flow vector} $f_{\cF}$ corresponding to path flow $\cF$, with
$
f_{\cF}(i)=\sum_{(P,\delta)\in \cF: P\ni i}{\delta}
$
for each link $i$, shows the flow along each link.

The \emph{multiflow routing and scheduling problem} is a covering problem, where given source-destination pairs with associated utilities, the goal is to find a path flow $\cF$ together with a fractional link schedule $\cS$ of length $1$, such that\footnote{Essentially, the schedule here gives a probability distribution over the feasible sets of links.} for each link $i$, the link flow is at most the link capacity provided by the schedule, $f_{\cF}(i) \le c_{\cS}(i)$, and the \emph{flow value}
\[
W=\sum_{i=1}^p \omega_i\cdot \sum_{(P_k,\delta_k)\in \cF, P_k\in \cP_i}{\delta_k}
\]
 is maximized.
Let us verify that this problem is also refinable. Consider a feasible solution in (the physical model) that consists of a path flow $\cF=\{(P_k,\delta_k): k=1,2,\dots\}$ and a schedule $\cS=\{(I_k,t_k) : k=1,2,\dots\}$ of length $\sum_{k\ge 1}t_k=1$, such that $f_{\cF}(i) \le c_{\cS}(i)$. As observed in the previous section, the schedule $\cS$ can be refined into a schedule $\cS'=\{(I_k^s,t_k)\}_{k,s}$ in $\cG(L)$, where $\cS'$ serves the same demand vector as $\cS$ does, and $\cS'$ has length at most $K=O(\log\log\Ds)$ times more than the length of $\cS$. Now we normalize the refined schedule to have length 1. Then, the following modified path flow $\cF'=\{(P_k,\delta_k/K): k=1,2,\dots\}$ together with the new schedule will be feasible in $\cG(L)$, as all link demands will be served. Moreover, the value of $\cF'$ is at most $K$ times that of $\cF$. Hence, the problem is refinable.

Thus, applying the constant factor approximation  algorithm of~\cite{wan14} for constant inductive independent conflict graphs (the result holds with unit utilities) gives an $O(\log\log\Ds)$-approximation for multiflow routing  and scheduling problem in the physical model.
It should also be noted that the fractional scheduling and routing and scheduling problems can be reduced to the {\wcapacity} problem using linear programming techniques (described e.g. in~\cite{jansen03}), as it was shown in~\cite{wan09}. We will further discuss this in Sec.~\ref{S:ratecontrol}.

\subparagraph*{Extensions to Multi-Channel Multi-Antenna Settings}
 All problems above may be naturally generalized to the case when there are several channels (e.g. frequency bands) available and moreover, wireless nodes are equipped with multiple antennas and can work in different channels simultaneously. We denote the setting with multiple antennas/channels as \emph{MC-MA}.

 It is easy to show that our refinement framework can be extended to MC-MA with very little change.
Assume that each node $u$ is equipped with $a(u)$ antennas numbered from $1$ to $a(u)$ and can use a set $\cC(u)$ of channels. Consider a link $i$ that corresponds to the pair of nodes $s_i$ and $r_i$. There are $a(s_i)a(r_i)|\cC(s_i)\cap \cC(r_i)|$ \emph{virtual} links corresponding to each selection of an antenna of the sender node $s_i$, an antenna of receiver node $r_i$ and a channel $c\in \cC(s_i)\cap \cC(r_i)$ available to both nodes.  Thus a virtual link is described by the tuple $(i, a_s, a_r, c)$, where $a_s$ ($a_r$) denotes the antenna index at $s_i$ ($r_i$, respectively), and $c$ denotes the channel. We call link $i$ \emph{the original} of its virtual links. Note that the formulation above can easily be generalized to the case where certain antennas don't work  in certain channels, e.g., due to multi-path fading.

A set of (virtual) links $S$ is feasible in MC-MA if and only if no two links in $S$ share an antenna (i.e., they do not use the same antenna of the same node), and for each channel $c$, the set of originals of links in $S$ using channel $c$ is feasible in the physical model.
Then, an $O(\log\log\Ds)$-tight refinement for the MC-MA physical model by a conflict graph can be found by a simple extension of the existing refinement for the single channel case to the virtual links (see Appendix~\ref{A:mcmr} for details).
This implies, in particular, that all scheduling problems considered in the previous sections can also be approximated in the MC-MA setting within an approximation factor $O(\log\log\Ds)$, as the corresponding approximations for the conflict graph hold with MC-MA~\cite{wan14}.

\section{Rate Control and Scheduling}\label{S:ratecontrol}

 The most important application of efficient approximation algorithms for scheduling problems with different thresholds is the application to scheduling with rate control. This is achieved first by obtaining a double-logarithmic approximation to {\wcapacity} with rate control. This will then lead to similar approximations for fractional scheduling and joint routing and scheduling problems.

\subparagraph*{{\capswcapacity} with Rate Control}
 By Shannon's theorem, given a set $S$ of links simultaneously transmitting in the same channel, the transmission rate $r(S,i)$ of a link $i$ is a function of $SIR(S,i)$. Thus, we consider the {\wcapacity} problem where each link $i$ has an associated \emph{utility function} $u^i:\mathbb{R}_+\rightarrow \mathbb{R}_+$, and the weight of link $i$ is the value of $u^i$ at $SIR(S,i)$ if link $i$ is selected in the set, and $0$ otherwise. As before, the goal is, given the links with utility functions,  to find a subset $S$ that maximizes the total weight $\sum_{i\in S} u^i(r(S,i))$. We assume that $u^i(SIR(S,i))=0$ if $SIR(S,i)<1$.

An $O(\log n)$-approximation for this variant of {\wcapacity} has been obtained in~\cite{KesselheimESA12}. We show that this can be replaced with $O(\log\log\Ds')$, where $\Ds'(L)=\max_{i,j\in L}\frac{u^i_{max}l_i}{u^j_{min}l_j}$ and $u^i_{min},u^i_{max}$ are the minimum and maximum possible utility values for the given instance and link. This is achieved by reducing the problem to {\wcapacity} in an extended instance.

Let us fix a utility function $u$. First, assume that the possible set of weights for each link is a discrete set $u_{min}=u_1<u_2<\cdots<u_{t}=u_{max}$. Then, we can replace each link $i$ with $t$ copies  $i_1,i_2,\cdots,i_t$ with different thresholds and fixed weights, where $\omega_{i_k}=u_k$ and $\beta_{i_k}=\min\{x: u^{i_k}(x) \ge u_k\}$ if $\beta_{i_k}\ge 1$ and $\omega_{i_k}=0$ otherwise. Now, the problem becomes a {\wcapacity} problem for the modified instance $L'$ with link replicas and fixed weights. Observe that no feasible set in $L'$ contains more than a single copy of the same link, as the copies occupy the same geometric place, implying that  each feasible set of the extended instance corresponds to a feasible set of the original instance, with an obvious transformation. The effective length diversity of the extended instance is  $\Ds(L')=\Ds'(L)$. Thus, using the approximation algorithm for the fixed rate {\wcapacity} problem, we  obtain an $O(\log\log \Ds'(L))$-approximation for {\wcapacity} with rate control.

For the case when the number of possible utility values is too large or the set is continuous, a standard trick can be applied. Let $u^i_{max},u^i_{min}$ be as before. The  extended instance $L'$ is constructed by replacing each link $i$ with $O(\log u^i_{max}/u^i_{min})$  copies $i_1,i_2,\dots$ of itself and assigning each replica $i_k$ weight $\omega_k=2^{k-1}$ and threshold $\beta_k=\min\{x : 2^{k-1} \le u^i(x) \le 2^{k}\}$ if $\beta_k\ge 1$ and let $\omega_k=0$ otherwise. It is easy to see that the optimum value of {\wcapacity} with fixed rates in $L'$ is again an $O(\log\log\Ds'(L))$-approximation to {\wcapacity} with rate control.

If the value $\log u^i_{max}/u^i_{min}$ is still too large, it may be inefficient to have $O(\log u^i_{max}/u^i_{min})$ copies for each link. It is another standard observation that only the last $O(\log n)$ copies of each link really matter, as restricting to only those links  degrades  approximation by a factor of at most 2.

\subparagraph*{Fractional Scheduling with Rate Control}
In this formulation, we redefine a fractional schedule to be  a set $\cS=\{(I_k,t_k) : k=1,2\dots,q\}\subseteq 2^L\times\mathbb{R}_+$, namely, $I_k$ are arbitrary subsets, rather than independent ones.  We redefine the link capacity vector ${\hat c}_{\cS}$ to incorporate the rates as follows:
\begin{equation}\label{E:ratecap}
{\hat c}_{\cS}(i) = \sum_{(I,t)\in \cS: I\ni i}t\cdot r(i,I).
\end{equation}
 The \emph{fractional scheduling with rate control} problem is to find a  minimum length schedule $\cS$ that serves a given demand vector $d$, namely, such that  for each link $i\in L$,
$
{\hat c}_{\cS}(i)\ge d(i).
$

The problem can be formulated as an exponential size linear program $LP_1$, as follows.
\begin{align}
\nonumber \min \sum_{I\subseteq L}{t_I}   \text{ subject to } & \sum_{I\subseteq L : I\ni i}t_I \cdot r(i,I) \ge d(i) && \forall i\in L\\
\nonumber & t_I\ge 0 && \forall I\subseteq L
\end{align}

The dual program $LP_2$ looks as follows:
\begin{align}
\nonumber \max \sum_{i\in L}{d(i) y_i}   \text{ subject to } & \sum_{i\in I}y_i\cdot  r(i,I) \ge 1 && \forall I\subseteq L\\
\nonumber & y_i\ge 0 && \forall i\in L
\end{align}

As~\cite[Thm. 5.1]{jansen03} states, if there is an approximation algorithm that finds a set $\hat I$ such that $\sum_{i\in \hat I}y_i r(i,\hat I)\ge \frac{1}{a}\max_{I\subseteq L}\sum_{i\in I}y_i r(i,I)$, then there is an $a$-approximation algorithm for $LP_1$, where the former algorithm acts as an approximate separation oracle for $LP_1$. But this auxiliary problem is simply a special case of the {\wcapacity} with rate control, which we can approximate within a double-logarithmic factor. Thus, there is an approximation preserving reduction from the fractional scheduling with rate control to {\wcapacity} with rate control. By the obtained approximation for {\wcapacity}, we obtain an $O(\log\log\Ds')$-approximation for fractional scheduling with rate control.

\subparagraph*{Routing, Scheduling and Rate Control}
The rate-control variant of the routing and scheduling problem is formulated in the same way as for the fixed rate setting, with the only modified constraint being the capacity constraints, which, instead of the link capacity vector $c_{\cS}$, now use the modified variant ${\hat c}_{\cS}$ that incorporates the link rates (see the definition in (\ref{E:ratecap})).

This problem can also be reduced to {\wcapacity} with rate control, using similar methods as for the fractional scheduling problem. The reduction is nearly identical to the reduction of fixed rate versions of these problems to {\wcapacity}, presented in~\cite[Thm. 4.1]{wan09}.

Thus, we can conclude that there is an $O(\log\log\Ds')$-approximation algorithm for joint routing, scheduling and rate control that uses {\wcapacity} with rate control as a subroutine.


\bibliographystyle{abbrv}
\bibliography{Bibliography}


\appendix

\newcounter{lastthm}
\setcounter{lastthm}{\value{theorem}}

\section{Simulation Results}\label{S:simulations}

The objective of these simulations is to see how the conflict graph approximation of the physical model behaves on probabilistically generated instances. First, we would like to see how do different parameter values affect the approximation performance. This includes fine-tuning the parameters $\gamma$ and $\delta$ in the conflict graph $\cG_\gamma^\delta$, as well as the power assignments $P_{\tau}$ predicted by Thm.~\ref{T:obliviouspowers}.

We tested the {\wcapacity} problem with uniform rates against some simple algorithms and heuristics.
We generated random link sets by placing $n$ links of lengths in the interval $(1,l_{max})$ in a $1000\times 1000$ square on the plane, where $l_{max}$ varies from 10 to 250. While the link positions and directions are chosen from uniform random distribution, the link lengths follow log-uniform distribution so that shorter links are more frequent than longer. Further, each link $i$ has fixed weight $\omega_i$ between 1 and 100, sampled from log-uniform distribution. The physical model parameters were set as $\alpha = 2.8$ and $\beta = 1.0$, and  the number of links was set to 400. The results are the averages over 20 instances.

We implemented the local ratio {\wcapacity} algorithm of~\cite{ackoglu} for conflict graphs $\cG_{\gamma}^\delta$ with $\delta = \delta_0 + \epsilon (1 - \delta_0)$ ($\delta_0$ given by Theorem~\ref{T:obliviouspowers}) and $\epsilon = 0.1$ and $0.9$. The power assignment for checking feasibility is as recommended by Lemmas~\ref{P:mainlemma1} and \ref{P:mainlemma2}. The best factor $\gamma$ is found by binary search. The first thing to notice from the figure is that, at least for randomly deployed instances, smaller values of $\delta$ are more efficient.

\begin{figure}[htb]
\includegraphics[width=1.0\textwidth]{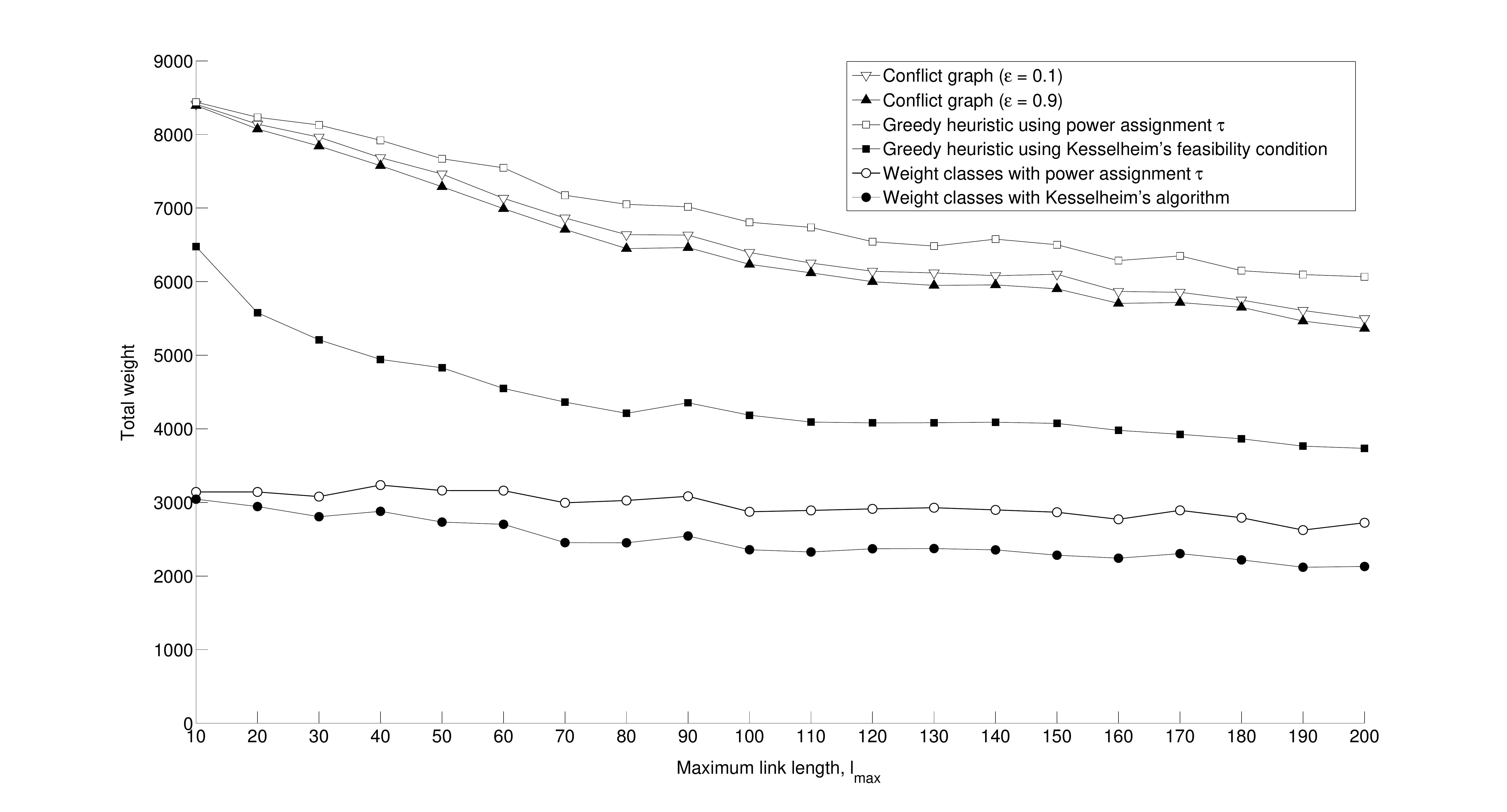}
\caption{Solutions of {\wcapacity} as a function of link diversity.}
\label{fig:weightedcapacity}
\end{figure}

These results can be compared with two other algorithms, based on uniform weight  {\wcapacity} algorithms.

The first algorithm is a greedy heuristic, which orders the links in increasing order of \emph{length by weight} ratio and accepts or rejects a link by simply checking feasibility each time. The first variant (top graph) checks feasibility using the same power assignment $P_\tau$ recommended by Lemmas~\ref{P:mainlemma1} and~\ref{P:mainlemma2}. The second variant (middle graph) checks feasibility based on a condition proposed by Kesselheim~\cite{kesselheimconstantfactor} (it should be noted, though, that this condition is used as it is presented in~\cite{kesselheimconstantfactor}; a more careful tuning of parameters may give better results). As it can be seen, this heuristic performs best on randomly generated instances, event though it might have poor worst-case performance.

 The second algorithm simply partitions the links into weight classes where the weight diversity is at most a factor of two, executes the uniform greedy {\wcapacity} algorithm on each subset and chooses the best solution obtained. This algorithm can be seen to perform rather poorly, although theoretically it has provable  logarithmic worst-case approximation.

Overall, it can be concluded that the conflict graph approximation not only gives theoretical worst-case guarantees, but also works fairly well on randomly generated instances. The experiments show that by fine-tuning the parameters for particular instances, good results can be obtained.

\section{Omitted Proofs: The Proof of Theorem~\ref{T:obliviouspowers}}

\setcounterref{theorem}{T:obliviouspowers}
\addtocounter{theorem}{-1}

\begin{theorem}
Let $\delta_0=\frac{\alpha-m+1}{2(\alpha-m) + 1}$. If $\delta\in (\delta_0,1)$ and the constant $\gamma>1$ is large enough, there is a value $\tau \in (0,1)$ such that each independent set in $\cG_{\gamma}^{\delta}$ is $P_{\tau}$-feasible.
\end{theorem}

Before going into details, note that by setting the power assignment $P_\tau$ in the feasibility formula, we have that a set $S$ of links is $P_\tau$-feasible if and only if for each link $i\in S$,
\[
I_{\tau}(S,i)=\sum_{j\in S\setminus \{i\}}\frac{\s_j^{\tau\alpha}\s_i^{(1-\tau)\alpha}}{d_{ji}^\alpha} < 1,
\]
where the additive operator $I_{\tau}$ is defined as follows: for all $i,j$, we set $I_\tau(j,i)=\frac{\s_j^{\tau\alpha}\s_i^{(1-\tau)\alpha}}{d_{ji}^\alpha}$, $I_{\tau}(i,i)=0$ and $I_{\tau}(S,i)=\sum_{j\in S}I_{\tau}(j,i)$.

The proof is a simple combination of Lemmas~\ref{P:mainlemma1} and~\ref{P:mainlemma2} below.
In order to bound $I_\tau(S,i)$ for a given set $S$ and a link $i$, we split $S$ into equilength subsets, bound the contribution of each subset separately (using Lemma~\ref{P:oblcore}), then combine the bounds into one.
The core of the proof, Lemma~\ref{P:oblcore}, is a careful application of a common packing argument.

\setcounter{theorem}{\value{lastthm}}

\begin{lemma}\label{P:oblcore}
Let $\mu, \tau \in (0,1)$ and $\gamma \ge 1$, let $S$ be an equilength set of links such that for all $j,k\in S$, $d(j,k) > \gamma \s_0$, where $\s_0=\min_{j\in S}\s_j$, and let link $i$ be such that $d(i,j) > \gamma \s_i^\mu\s_j^{1-\mu}$ for all $j\in S$.
Then,
 \[
 \displaystyle I_{\tau}(S,i)=
 O\left(\gamma^{-\alpha} \left(\frac{\s_i}{\s_0}\right)^{(1-\tau)\alpha - \mu(\alpha-m)}\cdot \min\left\{1,\frac{\s_i}{\s_0}\right\}^{-\mu} \right).
 \]
\end{lemma}
\begin{proof}
 First, let us split $S$ into two subsets $S'$ and $S''$ such that $S'$ contains the links of $S$ that are closer to $r_i$ than to $s_i$,
i.e., $S'=\{j\in S: \min\{d(s_j,r_i), d(r_j,r_i)\} \leq \min\{d(s_j,s_i), d(r_j,s_i)\}\}$ and $S''=S\setminus S'$. Let us consider the set $S'$ first.

 For each link $j\in S'$, let $p_j$ denote the  endpoint of $j$ that is closest to node $r_i$. Denote $q= (\s_i/\s_0)^{\mu}$.
Consider the subsets $S_1,S_2,\dots$ of $S'$, where
 $
 S_r=\{j\in S': d(j,i)=d(p_j,r_i)\leq \gamma(q \s_0+(r-1)\s_0)\}.
 $
 Note that $S_1$ is empty: for every $j\in S'$, $d(j,i)>\gamma q\s_0$. Thus, $S'=\cup_{r=2}^\infty{S_r}$. Let us fix an $r>1$.
For every pair of links $j,k\in S_r$, we have that
 $d(p_j,p_k) \ge d(j,k) > \gamma\s_0$  and that $d(p_j,r_i)\leq \gamma (q \s_0+(r-1)\s_0)$ for each $j\in S_r$ (by the definition of $S_r$), so using the doubling property of the metric space, we get the following bound:
\begin{equation}
|S_r|=|\{p_j\}_{j\in S_r}|\leq C\cdot \left(\frac{\gamma(q\s_0+(r-1)\s_0 )}{\gamma\s_0}\right)^{m} = C \left(q+r-1\right)^{m}.\label{E:strs}
\end{equation}
Note also that $\s_j \leq 2\s_0$ and $d(i,j) \ge \gamma(q\s_0+(r-2)\s_0)$ for every link $j\in S_r\setminus S_{r-1}$ with $r>1$; hence,
\begin{equation}
I_{\tau} (j, i) \le \frac{\s_j^{\tau\alpha} \s_i^{(1 - \tau)\alpha}}{d(i,j)^\alpha}
 \leq \left(\frac{\s_i}{\s_0}\right)^{(1-\tau)\alpha}\left(\frac{2\s_0}{\gamma(q\s_0+(r-2)\s_0)}\right)^\alpha =\frac{Z_i}{\left(q+r-2\right)^\alpha},\label{E:feqs}
\end{equation}
where $Z_i=(2/\gamma)^\alpha (\s_i/\s_0)^{(1-\tau)\alpha}$.
Recall that $S_{r-1}\subseteq S_r$ for all $r>1$, $S_1=\emptyset$ and $S'=\cup_{r=2}^\infty{S_r}$. Using~(\ref{E:strs}) and (\ref{E:feqs}), we have:
\begin{align*}
{I_{\tau}(S',i)} & = \sum_{r\geq 2}{\sum_{j\in S_r\setminus S_{r-1}}{I_{\tau}(j,i)}} \nonumber \\
& \leq \sum_{r\geq 2}{\left(|S_r|-|S_{r-1}|\right)\frac{Z_i}{\left(q+r-2\right)^\alpha}} \nonumber \\
& = Z_i \cdot \sum_{r\geq 2}{|S_r|\left( \frac{1}{\left(q+r-2\right)^\alpha} - \frac{1}{\left(q+r-1\right)^\alpha} \right)}\\
& \le Z_i \cdot \sum_{r\geq 2}{\frac{\alpha}{(q+r-1)^{\alpha-m+1}}}\\
&  = O(Z_i) \cdot\left(\frac{1}{q^{\alpha - m}} + \frac{1}{q^{\alpha - m + 1}}\right). 
\end{align*}
The proof for the set $S'$ now follows by plugging the values of $q$ and $Z_i$ in the expression above.

The proof holds symmetrically for the set $S''$, by swapping the senders with corresponding receivers in the argument.
\end{proof}

 \begin{lemma}\label{P:mainlemma1}
 Let $S$ be a set of links that is independent in $\cG_\gamma^\delta$ and let link $i$ be such that $\s_i\ge\max_{j\in S}\s_j$. Then for each $ \tau > 1- \frac{1+\delta}{2} (\alpha-m)/\alpha$,
 $
 I_{\tau} (S, i)= O\left(\gamma^{-\alpha/2}\right).
 $
 \end{lemma}
\begin{proof}
Let us split $S$ into equilength subsets $L_1, L_2,\dots$  with
$
L_t=\{j\in S: 2^{t-1}\s_0 \leq \s_j<2^t \s_0\},
$
 where $\s_0=\min_{j\in S}\{\s_j\}$. Let $\e_t=\min_{j\in L_t}\s_j$. The independence condition between $i$ and any other link $j\in S$ is $d_{ij}d_{ji} > \gamma \s_i^{1+\delta}\s_j^{1-\delta}$, which implies that $d(i,j) > \sqrt{\gamma}\s_i^{\frac{1+\delta}{2}}\s_j^{\frac{1-\delta}{2}}$. Similarly, $d(j,k) > \sqrt{\gamma}\e_t$ for all $j,k\in L_t$. By applying  Lemma~\ref{P:oblcore} with $\gamma=\sqrt{\gamma}$ and $\mu=(1+\delta)/2$, we obtain
\[
{I_{\tau}(L_t,i)} = O\left(\gamma^{-\alpha/2}\left(\frac{\e_t}{\s_i}\right)^{\mu(\alpha-m) - (1-\tau)\alpha }\right).
\]
Let us combine the bounds above into a geometric series:
\begin{align*}
{I_{\tau}(S,i)} &= \sum_{1}^{\infty}{I_{\tau}(L_t,i)} \\
&\le \frac{O(\gamma^{-\alpha/2})}{\s_i^{\mu(\alpha-m)  - (1-\tau)\alpha}}\sum_{t=0}^{\lceil\log{\s_i/\s_0}\rceil}{(2^{t}\s_0)^{\mu(\alpha-m)  - (1-\tau)\alpha}}.
\end{align*}
Recall that we assumed $\tau > 1- \mu (1-m/\alpha)$; hence, $\mu(\alpha-m)  - (1-\tau)\alpha> 0$.
Thus, the last sum is bounded by $O(\s_i^{(1-\tau)\alpha - \mu(\alpha-m)})$, which implies the lemma.
\end{proof}

 \begin{lemma}\label{P:mainlemma2}
Let $S$ be a set of links that is independent in $\cG_\gamma^\delta$ and let link $i$ be such that $\s_i\le \min_{j\in S}\s_j$. Then for each $\tau < 1- \frac{1-\delta}{2}(\alpha - m + 1)/\alpha$,
 $
 I_{\tau} (S, i)= O\left(\gamma^{-\alpha/2}\right).
 $
 \end{lemma}
 \begin{proof}
 Let us split $S$ into equilength subsets $L_1, L_2,\dots$, where
$
L_t=\{j\in S: 2^{t-1}\s_i \leq \s_j<2^t \s_i\}.
$
Let $\e_t=\min_{j\in L_t}\s_j$. Recall that $\e_t \ge 2^{t-1}\s_i$. From the independence condition, we have, as in the proof of Lemma~\ref{P:mainlemma1}, that $d(i,j) > \sqrt{\gamma}\s_j^{\frac{1+\delta}{2}}\s_i^{\frac{1-\delta}{2}}$ for each $j\in S$ (note the difference, though, as here $\s_i\le \s_j$) and $d(j,k)>\sqrt{\gamma}\e_t$ for all $j,k\in L_t$. We apply Lemma~\ref{P:oblcore} with $\gamma = \sqrt{\gamma}$ and $\mu=(1-\delta)/2$ to $L_t$ and link $i$ to get:
\begin{align*}
{I_{\tau}(L_t,i)} &= O\left(\gamma^{-\alpha/2}\left(\frac{\s_i}{\e_t}\right)^{(1-\tau)\alpha - \mu(\alpha-m+1)}\right)\\
&=O\left(\gamma^{-\alpha/2}\left(\frac{1}{2^{t-1}}\right)^{(1-\tau)\alpha - \mu(\alpha-m+1)}\right).
\end{align*}
Recall that we assumed $\tau < 1-\mu(\alpha- m + 1)/\alpha$, implying $\eta=(1-\tau)\alpha - \mu(\alpha-m+1) > 0$.
Thus, we have:
$
{I_{\tau}(L,i)}= \sum_{1}^{\infty}{{I_{\tau}(L_t,i)}}\leq \gamma^{-\alpha/2}\sum_{t=0}^{\infty}{\frac{1}{2^{\eta t}}} = O\left(\gamma^{-\alpha/2}\right).
$
\end{proof}

It is easily checked that when $\delta\in (\delta_0, 1)$,  the interval $(b,e)$ with $b=1-\frac{1+\delta}{2}\cdot \frac{\alpha-m}{\alpha}$ and $e=1-\frac{1-\delta}{2}\cdot\frac{\alpha - m + 1}{\alpha}$ is non-empty, and $\tau$ can be chosen to be any point in $(b,e)$. This completes the proof of Thm.~\ref{T:obliviouspowers}.

\section{Omitted Proofs: Lemmas for Thm.~\ref{T:sandwich}}

\begin{lemma}\label{P:longlink}
For each pair $i,j$ of links, $l_jd(i,j) \le 2d_{ij}d_{ji} + l_il_j$.
\end{lemma}
\begin{proof}
Note that if $d(i,j) > l_j$ then the claim holds trivially, so assume that $d(i,j)\le l_j$. The triangle inequality implies: $l_j\le d_{ij} + d_{ji} + l_i$. Multiplying both sides by $d(i,j)$ yields the claim: $d(i,j)l_j \le d(i,j)d_{ij} + d(i,j)d_{ji} + d(i,j)l_i\le d_{ji}d_{ij} + d_{ij}d_{ji} + l_jl_i$.
\end{proof}

\begin{lemma}\label{P:trapez}
For each triple $i,j,k$ of links with $l_i\le l_j \le l_k$,
\begin{align*}
\min\{d_{jk},d_{kj}\} &\le d(i,j) + l_i + l_j + d(i,k),\\
\max\{d_{jk},d_{kj}\} &\le d(i,j) + l_i + l_j + l_k + d(i,k).
\end{align*}
\end{lemma}
\begin{proof}
The proof follows from the triangle inequality. Let $a,b$ be nodes of link $i$ (possibly coinciding) and $c,d$ be nodes of $j$ and $k$ respectively,  such that $d(i,j) = d(a,c)$ and $d(i,k)=d(b,d)$. Let $e$ ($f$) be the remaining node of link $j$ ($k$, resp.). Then the lemma follows by applying the triangle inequality to the (only) two possible cases: 1. $\{d_{jk},d_{kj}\}=\{d(c,d),d(e,f)\}$,  2. $\{d_{jk},d_{kj}\}=\{d(c,f),d(e,d)\}$.
\end{proof}

\setcounterref{theorem}{P:triangles}
\addtocounter{theorem}{-1}

\begin{lemma}
Let $i,j,k$ be such that $\s_i \le \s_j\le \s_k$ and $i$ is $f$-adjacent with both $j$ and $k$, where $f$ is a non-decreasing sublinear function. Then
\begin{align*}
d_{jk}d_{kj} < &18\s_i\s_kf(\s_k/\s_i) + 13\s_j\s_k \\
&+ 2\s_j\sqrt{\s_i\s_kf(\s_k/\s_i)} + \s_k\sqrt{\s_i\s_jf(\s_j/\s_i)}.
\end{align*}
\end{lemma}
\begin{proof}
Since $i,j$ and $i,k$ are $f$-adjacent, we have:
	\begin{equation}
	\label{E:indineq1} d_{ij}d_{ji} \le \s_i\s_j f\left(\frac{\s_j}{\s_i}\right),\text{ and }
	 d_{ik}d_{ki} \le \s_i\s_k f\left(\frac{\s_k}{\s_i}\right),
	\end{equation}
	Using Lemma~\ref{P:trapez} we can write:
	\begin{align*}
	d_{jk}d_{kj}\le &(d(i,j)+l_i+l_j+d(i,k))^2 \\
	&+ l_k (d(i,j) + l_i+l_j + d(i,k)).
	\end{align*}
	We bound each term separately, starting from the first square. The lemma follows by simply adding up the bounds.
	\begin{itemize}
	\item{$d(i,k)^2\le \s_i\s_kf(\s_k/\s_i)$, \\
	$d(i,j)^2\le \s_i\s_jf(\s_j/\s_i)\le \s_i\s_kf(\s_k/\s_i)$ and \\
	$2d(i,j)d(i,k)\le 2\s_i\s_kf(\s_k/\s_i)$, follow from~(\ref{E:indineq1}) and monotonicity of $f$,}
	\item{$l_i^2 + l_j^2 + 2l_il_j \le 4\s_j\s_k$, from $\s_i\le \s_j\le \s_k,$}

	\item{$2(l_id(i,j) + l_jd(i,j) + l_id(i,k))$\\
	$\le 4(d_{ij}d_{ji} + d_{ij}d_{ji}+ d_{ik}d_{ki}) + 2l_i(2l_j+l_k)$\\
	$\le 12\s_i\s_kf(\s_k/\s_i) + 6\s_i\s_k$, from Lemma~\ref{P:longlink}, eqs. (\ref{E:indineq1}) and monotonicity of $f$,}
	\item{$2l_jd(i,k) \le 2\s_j\sqrt{\s_i\s_kf(\s_k/\s_i)}$, from~(\ref{E:indineq1}),}
	\item{$l_kd(i,j)\le \s_k\sqrt{\s_i\s_jf(\s_j/\s_i)}$, from~(\ref{E:indineq1}),}
	\item{$l_k(l_i+l_j)\le 2\s_k\s_j$, from $\s_i\le \s_j\le \s_k$,}
	\item{$l_kd(i,k)\le 2d_{ik}d_{ki} + l_il_k\le 2\s_i\s_kf(\s_k/\s_i) + \s_i\s_k$, from Lemma~\ref{P:longlink} and~(\ref{E:indineq1}).}
	\end{itemize}
\end{proof}

\begin{lemma}
 Let $i$ be a link and $\rho>1$.
If $E$ is a $\rho$-independent set of links where each link $j\in E$ is $f$-adjacent with $i$ and satisfies $\s_i\le \s_j \le c\s_i$ for a constant $c$, then $|E|=O(1)$.
\end{lemma}
\begin{proof}
	Take subsets $A=\{j\in E : d_{ji}=\min\{d_{ij},d_{ji}\}\}$ and $B=E\setminus A$. We bound the size of $A$, the case of $B$ being symmetric. Consider the distances of links in $A$ to the node $r_i$. From $f$-adjacency assumption and the definition of $A$, we have that for each $j\in A$, $d(s_j,r_i)\le \sqrt{d_{ij}d_{ji}} \le \sqrt{\s_i\s_jf(\frac{\s_j}{\s_i})}\le \sqrt{hc}\cdot \s_i$, where constant $h$ is such that $f(x)\le hx$. On the other hand, we
	claim that $d(s_j,s_k) > (\sqrt{\rho}-1)\s_j \ge (\sqrt{\rho}-1)\s_i$ for all links $j,k\in S$ with $\s_j\le \s_k$. Indeed, if we assume the contrary, then the triangle inequality gives: $d_{kj}\le d(s_k,s_j)+l_j\le (\sqrt{\rho}-1)\s_j + \s_j=\sqrt{\rho}\s_j$ and $d_{jk}\le d(s_j,s_k) + l_k \le (\sqrt{\rho}-1)\s_j + \s_k \le \sqrt{\rho}\s_k$, contradicting $\rho$-independence.

	Thus, the distance from every point in $A$ to $r_i$ is at most $\sqrt{hc}\s_i$ and the mutual distances of points in $A$ are at least $(\sqrt{\rho}-1)\s_i$, so the doubling property of the metric space gives:
 $|A|=O((\sqrt{hc/\rho})^m)= O(1)$, which completes the proof.
\end{proof}

\section{The Proof of Theorem~\ref{T:inductiveindep}}

\setcounterref{theorem}{T:inductiveindep}
\addtocounter{theorem}{-1}

\begin{theorem}
Let $f$ be a non-decreasing strongly sub-linear function with $f(x)\ge 40$ for all $x\ge 1$.  For every set $L$, the graph $\cG_f(L)$  is constant inductive independent.
\end{theorem}
\begin{proof}
We will show that $\cG_f(L)$ is constant inductive independent with respect to the ordering of links in a non-decreasing order by effective length (ties broken arbitrarily).
Fix  a link $i\in L$ and let $T$ be an $f$-independent set of links that are $f$-adjacent with $i$ and $\s_j\ge\s_i$ for all $j\in T$. It is sufficient to show that $|T|= O(1)$.

We will show that all links in $T$ except perhaps only one, are such that $\s_j \le c\s_i$ for some constant $c$: the proof then is completed by applying Lemma~\ref{P:simpleset}. The constant $c$ is determined by the properties of function $f$. Consider arbitrary two links $i,j\in T$. By Lemma~\ref{P:triangles}, and the assumption that $f(x)=o(x)$, we have that
$
d_{jk}d_{kj} \le 18\s_i\s_kf(\s_k/\s_i) + \s_j\s_k(13 + o(\s_j/\s_i)).
$
If $\s_j > c_1\s_i$ with large enough constant $c_1$ (depending on $f$), then the $o$-term is smaller than, say, $7$.  On the other hand, $f$-independence of $j$ and $k$ implies: $d_{jk}d_{kj} > \s_j\s_kf(\s_k/\s_j)$. Putting these together, we  obtain that for all links $j,k$ with $c_1\s_i < \s_j\le \s_k$,
$
\s_j\s_k f(\s_k/\s_j) \le 18\s_i\s_kf(\s_k/\s_i) + 20\s_j\s_k.
$
Since $f(\s_k/\s_j) \ge 40$, a simple manipulation gives us: $\frac{1}{2}\s_j\s_k f(\s_k/\s_j) \le 18\s_i\s_kf(\s_k/\s_i)$ or
$
f(y)/y \le 36\cdot f(x)/x,
$
where we denote $x=\s_k/\s_i$ and $y=\s_k/\s_j$. Strong sub-linearity of $f$ implies that there is a constant $c_2$ such that if $x > c_2 y$, then $36f(x)/x < f(y)/y$, so in virtue of the inequality above,  we must have that $\s_j/\s_i = x/y \le c_2$. Thus, we have proved that for all links $j\in T$ except maybe one, $\s_j\le \max\{c_1,c_2\} \s_i$, which completes the proof.
\end{proof}

\subsection{A Conflict Graph Refinement for MC-MA}\label{A:mcmr}

In order to find a conflict graph refinement for the MC-MA setting, it is sufficient to extend the existing refinement for the single channel case to the virtual links.

Let $L$ denote the set of virtual links and $L_o$ denote the corresponding originals. Let $\cG(L_o)$ be the conflict graph refinement of $L_o$, when assuming a single channel. We define the refinement graph $\cGm(L)$: the set of vertices of $\cGm(L)$ is the set of virtual links, and two virtual links are adjacent if at least one of the following holds: 1. they share an antenna, 2. they share a channel \emph{and} their originals are adjacent in $\cG(L_o)$, i.e., in the single channel setting.

Clearly, each independent set in $\cGm(L)$ is feasible. Since each feasible set in MC-MA is a collection of feasible sets of virtual links (in the sense of ordinary physical model) corresponding to different channels, the $O(\log\log\Ds)$-tightness of $\cG(L_o)$ easily implies $O(\log\log\Ds)$-tightness of $\cGm(L)$.
In order to prove inductive independence, it is enough to note that each virtual link has the same neighborhood in $\cGm(L)$ as in $\cG(L_o)$ (when translated back to ``originals''), except for two additional sets of links: $S$ -- the ones using the same antenna as the  sender  and $R$ -- the ones using the same antenna as the receiver. Note that the sets $S$ and $R$ are cliques in $\cGm(L)$, as they all share an antenna. Thus, each independent set in the neighborhood of a virtual link consists of unique replicas of an independent set in $\cG(L_o)$, plus at most two more links, one from $S$ and another from $R$. This readily implies that if $\cG(L_o)$ is $k$-inductive independent, then $\cGm(L)$ is $k+2$-inductive independent.

Thus, $\cGm$ is an $O(\log\log\Ds)$-tight refinement for the MC-MA physical model.

\end{document}